\documentclass[runningheads,a4paper]{llncs}

\usepackage{amssymb}
\usepackage{wrapfig}
\usepackage{graphicx}
\usepackage{cite}
\usepackage{url}
\usepackage{amssymb}
\usepackage{amsmath}
\usepackage{epstopdf}
\usepackage{bm}
\usepackage{multirow}
\usepackage{subcaption}
\usepackage{algpseudocode}
\usepackage{enumitem}
\usepackage[linesnumbered, lined,boxed,commentsnumbered, ruled,norelsize]{algorithm2e}
\usepackage{hyperref}

\urldef{\mailsa}\path||
\urldef{\mailsb}\path||
\urldef{\mailsc}\path||    
\newcommand{\keywords}[1]{\par\addvspace\baselineskip
\noindent\keywordname\enspace\ignorespaces#1}

\begin{document}

\title{A Generic Message Authentication Code: A combination of the Inter MAC and Carter-Wegman MAC}

\author{Chi Tran}
\institute{CTRSEC\\
\email{research@ctrsec.io}}

\titlerunning{Generic MAC}

\maketitle

\begin{abstract}
Message Authentication Code (MAC) is a method for providing integrity and authenticity assurances on the message by allowing the receiver to detect any changes to the message content. In this paper, we present a generic MAC named InterCW-MAC which can prevent replay attack and can deal with a untrusted receiver who may search the secret keys of the sender.
  
\end{abstract}

\section{Introduction}

Consider a scenario in which a sender owns a message $M$ and wants to send $M$ to a receiver. Because $M$ can be modified or loss during being transmitted, the receiver would like to check whether $M$ is intact or not.
A MAC is proposed to provide integrity and authenticity assurances on the message by allowing the receiver (who also possess the shared secret key of the sender) to detect any changes to the message content. A MAC consists of a tuple of algorithms ($\mathsf{KeyGen}$, $\mathsf{Tag}$, $\mathsf{Verify}$) as follows:
\begin{itemize}
\item $\mathsf{KeyGen}(1^\lambda) \rightarrow k$: The sender runs this algorithm which inputs a security parameter $\lambda$, and outputs a secret key $k$. The sender then sends $k$ to the receiver via a secure channel.
\item $\mathsf{Tag}(M, k) \rightarrow t$: The sender runs this algorithm which inputs $k$ and a message $M$, and outputs a tag $t$. 
\item $\mathsf{Verify} (M, t, k) \rightarrow \{0, 1\}$: The receiver runs this algorithm which inputs $M$, $t$ and $k$, and outputs 1 if $t$ is a valid tag and 0 otherwise.
\end{itemize}

\begin{definition}
A MAC is an additive homomorphic MAC if it satisfies:
\begin{equation}
Tag(M_1, k) + Tag(M_2, k) = Tag(M_1 +M_2, k)
\end{equation}
\end{definition}

\begin{definition}
A MAC is a multiplicative homomorphic MAC if it satisfies:
\begin{equation}
Tag(M_1, k) \cdot Tag(M_2, k) = Tag(M_1 \cdot M_2, k)
\end{equation}
\end{definition}

\subsection{Inner-product MAC}
The simplest additive MAC is the inner-product MAC. This MAC consists of the following algorithms:
\begin{itemize}
\item $\mathsf{KeyGen}(1^\lambda) \rightarrow k$: The sender runs this algorithm which inputs a security parameter $\lambda$, and outputs a secret key $k \in \mathbb{F}_q^{\xi}$ where $\mathbb{F}_q^{\xi}$ denotes a $\xi$-dimensional finite field $\mathbb{F}$ of a prime order $q$. The sender then shares $k$ to the receiver via a secure channel.
\item $\mathsf{Tag}(M, k) \rightarrow t$: The sender runs this algorithm which inputs $k \in \mathbb{F}_q^{\xi}$ and a message $M \in \mathbb{F}_q^{\xi}$, and outputs a tag $t$ such that:
\begin{equation}
t= M \cdot k \in \mathbb{F}_q
\end{equation}
The sender then transmits $\{M, t\}$ to the receiver.
\item $\mathsf{Verify} (M, t, k) \rightarrow \{0, 1\}$: The receiver runs this algorithm which inputs $\{M, t, k\}$ and checks if:
\begin{equation}
t \stackrel{?}{=} M \cdot k
\end{equation}
The algorithms outputs 1 if the equality holds ($t$ is a valid tag) and outputs 0 otherwise.
\end{itemize}

\begin{theorem}
The inner-product MAC is a additive homomorphic MAC.
\end{theorem}

\begin{proof}

\begin{itemize}
\item $Tag(M_1, k) =M_1 \cdot k$
\item $Tag(M_2, k) =M_2 \cdot k$
\item $Tag(M_1, k) + Tag(M_2, k) = (M_1 + M_2) \cdot k = Tag (M_1 + M_2, k)$ \qed
\end{itemize}
\end{proof}

\subsubsection{Security Analysis.} 
The inner-product MAC is secured from the brute forge search if $p$ is chosen large enough. This is because $k \in \mathbb{F}_q^{\xi}$, the probability to find $k$ via a brute force search is $1/q^{\xi}$. If $q$ is chosen large enough (i.e., 160 bits), the probability is $1/(2^{160})^{\xi}$, which is negligible.  

However, this MAC is not secured from the \emph{replay attack} because:
\begin{itemize}
\item In the first transmission: 
\begin{itemize}
\item The sender sends $\{M, t = M \cdot k\}$ to the receiver. 
\item The attacker captures $M$ and $t$ when they are transmitted.
\item The receiver verifies iff: $t = M \cdot k$. Suppose that the equality holds ($M$ is intact), the receiver then outputs 1.
\end{itemize}
\item In the next transmission:
\begin{itemize}
\item The sender sends $\{M', t' = M' \cdot k\}$ to the receiver.
\item The attacker re-uses the old message $M$ and the old tag $t$. It is clear that:
\begin{equation}
\label{equa:innerproductMAC}
\frac{t}{M} = \frac{t'}{M'}
\end{equation}
\item The attacker drops $\{M', t'\}$ and sends $\{M, t\}$ to the receiver.
\item The receiver verifies iff: $t = M \cdot k$. This equality will hold because of Eq. \ref{equa:innerproductMAC} but the receiver cannot know that the message was replaced.
\end{itemize}
\end{itemize}

\subsection{Carter-Wegman MAC}

To address replay attack which is a drawback of the inner-product MAC, the Wegman-Carten MAC \cite{CWMAC-STOC1977, CWMAC-Journal1979} has been proposed with an additional pseudorandom function. This MAC consists of the following algorithms:

\begin{itemize}
\item $\mathsf{KeyGen}(1^\lambda) \rightarrow \{k, k'\}$: The sender runs this algorithm which inputs a security parameter $\lambda$, and outputs a secret key $k \in \mathbb{F}_q^{\xi}$ which is used for tagging the message $M  \in \mathbb{F}_q^{\xi}$, a secret key $k' \in \mathcal{K'}$ which is used for permuting the tag. After that, the sender shares $\{k, k'\}$ to the receiver via a secure channel.

\item $\mathsf{Tag}(M, k, k') \rightarrow t$: The sender runs this algorithm which inputs the message $M$ and $\{k, k'\}$, and outputs a tag $t$ such that:
\begin{equation}
t = M \cdot k + f_{k'}(r) \in \mathbb{F}_q
\end{equation}
where $r$ denotes a random value and $f$ denotes a pseudorandom function such that $f: \mathcal{K'} \times \mathcal{R} \rightarrow \mathbb{F}_q$ ($\mathcal{K'}$ is the space of $k'$ and $\mathcal{R}$ is the space of $r$). The sender then transmits $\{M, t\}$ to the receiver. Note that (i) $f$ can be public, and (ii) $r$ is re-generated very transmission and is transmitted to the receiver via a secure channel like the keys (or is chosen large enough).
\item $\mathsf{Verify}(M, t, k, k', r, f) \rightarrow \{0, 1\}$: The receiver runs this algorithm which inputs $\{M, t, k, k', r, f\}$ and checks if:
\begin{equation}
t \stackrel{?}{=} M \cdot k + f_{k'}(r)
\end{equation}
The algorithm outputs 1 if the equality holds ($t$ is a valid tag), and outputs 0 otherwise.
\end{itemize}

\subsubsection{Security Analysis.}
Similar to the inner-product MAC, the keys of the Wegman-Carten MAC is not disclosed from the brute force search if $q$ and $\mathcal{K'}$ are chosen large enough. Furthermore, the Wegman-Carten MAC is secured from the replay attack:
\begin{itemize}
\item In the first transmission: 
\begin{itemize}
\item The sender sends $\{M, t = M \cdot k + f_{k'}(r)\}$ to the receiver. 
\item The attacker captures $M$ and $t$ when they are transmitted.
\item The receiver verifies iff: $t = M \cdot k + f_{k'}(r)$. Suppose that the equality holds (which means that $M$ is intact), the receiver then outputs 1.
\end{itemize}
\item In the next transmission:
\begin{itemize}
\item The sender sends $\{M', t' = M' \cdot k + f_{k'}(r')\}$ to the receiver.
\item The attacker re-uses the old message $M$ and the old tag $t$. It is clear that:
\begin{equation}
\label{equa:CW-MAC}
\frac{t}{M} \neq \frac{t'}{M'}
\end{equation}
\item The attacker drops $\{M', t'\}$ and sends $\{M, t\}$ to the receiver.
\item The receiver verifies iff: $t = M \cdot k$. This equality will not hold because of Eq. \ref{equa:CW-MAC}. The receiver then outputs 0. 
\end{itemize}
\end{itemize}

\subsection{Inter MAC}

Most of the MACs consider a scenario in which the sender sends a message along with the corresponding tag to the receiver. The receiver can check whether the message is intact using a shared key with the sender. However, the receiver can be untrusted (the sender should not share his/her secret key to the receiver). To address this problem, the inter MAC has been proposed~\cite{InterMAC-INFOCOM2012, Elar-1, COCOON, IEICE, MD-POR} as follows:

\begin{itemize}
\item $\mathsf{KeyGen}(1^\lambda, M) \rightarrow \{k, k'\}$: The sender runs this algorithm which inputs a security parameter $\lambda$ and a message $M \in \mathbb{F}_q^{\xi}$, and outputs secret keys $\{k \in \mathbb{F}_q^{\xi}, k' \in \mathbb{F}_q^{\xi}\}$ where $M \cdot k' = 0$. The sender then computes $k'' = k+ k'$, and sends $k''$ to the receiver.

\item $\mathsf{Tag}(M, k) \rightarrow t$: The sender runs this algorithm which inputs the message $M \in \mathbb{F}_q^{\xi} $ and the secret key $k \in \mathbb{F}_q^{\xi}$, and outputs a tag $t$ such that:
\begin{equation}
t= M \cdot k  \in \mathbb{F}_q
\end{equation}
The sender then transmits $\{M, t\}$ to the receiver.
\item $\mathsf{Verify} (M, t, k'') \rightarrow \{0, 1\}$: The receiver runs this algorithm which inputs $\{M, t, k''\}$ and checks if:
\begin{equation}
t \stackrel{?}{=} M \cdot (k + k')
\end{equation}
The algorithm outputs 1 if the equality holds ($t$ is a valid tag) and outputs 0 otherwise. 
\end{itemize}

\begin{theorem}
The inter MAC is a additive homomorphic MAC.
\end{theorem}

\begin{proof}

\begin{itemize}
\item $Tag(M_1, k+k') =M_1 \cdot (k+k')$
\item $Tag(M_2, k+k') =M_2 \cdot (k+k')$
\item $Tag(M_1, k+k') + Tag(M_2, k+k') = (M_1 + M_2) \cdot (k+k') = Tag (M_1 + M_2, k+k')$ \qed
\end{itemize}
\end{proof}

\subsubsection{Security Analysis.}
Given $k''$, which is the summation $(k + k')$,  the receiver cannot obtain the secret keys $k$ and $k'$ of the sender. This is because $k, k' \in \mathbb{F}_q^{\xi}$, the probability for an attacker to search $k$ (or $k'$) via a brute force search and then to obtain $k' = k'' - k$ (or $k = k'' - k'$) is $1/q^{\xi}$. If $q$ is chosen large enough (i.e., 160 bits), the probability to find $k$ and $k'$ is $1/2^{160}$, which is negligible. However, similar to the inner-product MAC, this inter MAC cannot be secured from replay attack.

\section{The Proposed InterCW-MAC} 
The InterCW-MAC is a combination between the Carter-Wegman MAC and the inter MAC. The InterCW-MAC is proposed to deal with two problems of the Carter-Wegman MAC and the inter MAC: preventing the replay attack and preventing the untrusted receiver to learn the secret key of the sender. 

\subsection{Construction}
The InterCW-MAC consists of the following algorithms:
\begin{itemize}
\item $\mathsf{KeyGen}(1^\lambda, M) \rightarrow \{k_1, k_1', k_2\}$: The sender runs this algorithm which inputs a security parameter $\lambda$, and outputs a secret key $k_1, k'_1 \in \mathbb{F}_q^{\xi}$ which is used for tagging the message $M$, and a key $k_2 \in \mathcal{K'}$ which is used for permuting the tag. The sender then computes $k''_1 = k_1 + k'_1$. The sender sends $k''_1$ and $k_2$ to the receiver via a secure channel.
\item $\mathsf{Tag}(M, k_1, k_2) \rightarrow t$: The sender runs this algorithm which inputs a message $M$ and $\{k_1, k_2\}$, and outputs a tag $t$ such that:
\begin{equation}
t = M \cdot k_1 + f_{k_2}(r) \in \mathbb{F}_q
\end{equation}
where $r$ denotes a random value and $f$ denotes a pseudorandom function such that $f: \mathcal{K'} \times \mathcal{R} \rightarrow \mathbb{F}_q$ ($\mathcal{K'}$ is the space of $k'$ and $\mathcal{R}$ is the space of $r$). The sender then transmits $\{M, t\}$ to the receiver. Note that (i) $f$ can be public, and (ii) $r$ is re-generated very transmission and is transmitted to the receiver via a secure channel like the keys (or is chosen large enough).
\item $\mathsf{Verify}(M, t, k''_1, k_2, r, f) \rightarrow \{0, 1\}$: The receiver runs this algorithm which inputs $\{M, t, k''_1, k_2, r, f\}$ and checks if:
\begin{equation}
t \stackrel{?}{=} M \cdot k''_1 + f_{k_2}(r)
\end{equation}
The algorithm outputs 1 if the equality holds ($t$ is a valid tag), and outputs 0 otherwise.
\end{itemize}

\subsection{Application of InterCW-MAC} 
We believe the InterCW-MAC can be applied in several scenarios. In this section, we describe an application of the InterCW-MAC in network coding-based distributed storage system. 

\subsubsection{Network coding-based distributed storage system.}
Network coding has been applied to distributed storage system \cite{NC-POR-IEEE2010, NC-POR-CCSW2010, NC-POR-IEEE2014, NC-POR-NetCod2012}. In this scenario, the system model consists of two types of entities: a client (trusted) and servers (untrusted). Suppose that a client owns an original file $F$ which consists of $m$ file blocks: $F = v_1|| \cdots ||v_m$. $v_i \in \mathbb{F}_q^{\xi}$ where $i\in \{1, \cdots, m\}$. The client wants to store redundantly encoded blocks in the servers in a way that the client can reconstructs the original file $F$ and can repair the encoded blocks in a corrupted server. From these file blocks, the client firstly creates $m$ augmented blocks $\{w_1, \cdots, w_m\}$ in which $w_i \in \mathbb{F}^{\xi+m}_q$ where $i \in \{1, \cdots, m\}$ has the form as follows:
\begin{equation}
\label{equation:networkcoding-fundamentalconcept}
w_{i} = (v_i, \overbrace{\underbrace{0, \cdots, 0, 1}_{i}, 0, \cdots, 0}^m) \in \mathbb{F}^{\xi+m}_q
\end{equation}

The client then randomly chooses $m$ coding coefficients $\alpha_1, \cdots, \alpha_m \stackrel{rand}{\leftarrow} \mathbb{F}_q$ and computes coded blocks using the linear combination as follows:
\begin{equation}
c=\sum_{i=1}^m \alpha_i \cdot w_i \in \mathbb{F}^{\xi+m}_q
\end{equation}

The coded blocks are then stored these coded blocks in the servers. To reconstruct the original file $F$, any $m$ coded blocks are required to solve $m$ augmented blocks $w_1, \cdots, w_m$ using the accumulated coefficients contained in the last $m$ coordinates of each coded block. After these $m$ augmented blocks are solved, $m$ file blocks $v_1, \cdots,v_m$ are obtained from the first coordinate of each augmented block. Finally, the original file $F$ is reconstructed by concatenating the file blocks. Note that the matrix consisting of the coefficients used to construct any $m$ coded blocks should have full rank. Koetter et al. \cite{NC-JournalTON2003} proved that if the prime $q$ is chosen large enough and the coefficients are chosen randomly, the probability for the matrix having full rank is high.

\subsubsection{Checking the data stored in the servers.}
Because the servers may be untrusted, the client must check the servers periodically to ensure that his/her data stored in the servers is always available and intact. In this case, the client can use the same secret keys to tag and to verify. 

However, when the client does not want to be burdened in checking the servers periodically, the client can delegate this task to another entity called \emph{third party verifier} (\emph{verifier} for short), which may be also untrusted. This entity is supposed to not collude with the servers. To deal with this scenario, we can apply the InterCW-MAC. 

\subsubsection{How to apply InterCW-MAC.} 
The InterCW-MAC can be applied to the above scenario as follows: 

\paragraph{a)}
$\mathsf{KeyGen}(1^\lambda, \{w_1, \cdots, w_m\}) \rightarrow \{k_1, k'_1, k_2\}$: The client runs this algorithm which inputs a security parameter $\lambda$ and the set of $m$ augmented blocks $\{w_1, \cdots, w_m\}$, and outputs secret keys $\{k_1, k_2\}$ as follows:
\begin{itemize}
\item $k_1, k_2 \stackrel{rand}{\leftarrow} \mathbb{F}^{\xi+m}_q$.
\item $k'_1 \in \mathbb{F}^{\xi+m}_q$ such that $w_i \cdot k'_1 = 0$ for all $i \in \{1, \cdots, m\}$.
\end{itemize}

The client then computes $k''_1 = k_1 +k'_1$. The client sends $k''_1$ and $k_2$ to the verifier.

The $\mathsf{KeyGen}$ introduces a challenge that how to generate $k_1'$ such that it is orthogonal to all $m$ augmented blocks. Formally, $k'_1 \cdot w_i = 0$ for all $i \in \{1, \cdots, m\}$. The algorithm to generate $k'$ is given as follows.
\begin{itemize}
\item $\mathsf{OrthogonalGen \text{--} SS}$ $(w_1, \cdots, w_m)$ $\rightarrow$ $k'_1$:
\begin{itemize}
\item Find the span $\pi$ of $w_1, \cdots, w_m \in \mathbb{F}_q^{\xi+m}$.
\item Construct the matrix $M$ in which $\{w_1, \cdots, w_m\}$ are the rows of $M$.
\item Find the null-space of $M$, denoted by $\pi^{\perp}_M$, which is the set of all vectors $u \in \mathbb{F}^{\xi+m}_q$ such that $M\cdot u^T = 0$.
\item Find the basis vectors of $\pi^{\perp}_M$, denoted by $b_1, \cdots, b_{\xi} \in \mathbb{F}^{\xi+m}_q$ // Theorem \ref{theorem:intermac1} will explain why the number of the basis vectors is $\xi$.
\item Let $B = \{b_1, \cdots, b_{\xi}\}$
\item Compute $k'_1 \leftarrow \mathsf{Kg \text{--} SS}(B)$.
\end{itemize}

\item  $\mathsf{Kg} (B = \{b_1, \cdots, b_{\xi}\}) \rightarrow k'$: this is the sub-algorithm used in $\mathsf{OrthogonalGen \text{--} SS}$
\begin{itemize}
\item Let $f$ be a Pseudorandom function such that $\mathcal{K} \times [1, \xi] \rightarrow \mathbb{F}_q$.
\item Generate $r_x \leftarrow f(k_{PRF}, x) \in \mathbb{F}_q, \forall x \in \{1, \cdots, z\}$ where $k_{PRF} \in \mathcal{K}$.
\item Compute $k'_1 \leftarrow \sum^z_{x = 1} r_x \cdot b_x \in \mathbb{F}_q^{\xi+m}$.
\end{itemize}
\end{itemize}

\begin{theorem}
\label{theorem:intermac1}
Given $\{w_1, \cdots, w_m\} \in \mathbb{F}_q^{\xi+m}$, the number of basis vectors of $\pi^{\perp}_M$ is $\xi$.
\end{theorem}

\begin{proof}
$\mathsf{rank}(M) = m$. Let $\pi_M$ be the space spanned by the rows of $M$. For any $m \times (\xi+m)$ matrix, the rank-nullity theorem gives: 
\begin{equation}
\mathsf{rank}(M) + \mathsf{nullity}(M) = \xi +m
\end{equation}
where $\mathsf{nullity}(M)$ is the dimension of $\pi^{\perp}_M$. Therefore,
\begin{equation}
\mathsf{dim}(\pi^{\perp}_M) = (\xi +m) - m = \xi
\end{equation}
In other words, the number of basis vectors of $\pi^{\perp}_M$ is $\xi$. In the $\mathsf{OrthogonalGen}$, we denoted the basis vectors by $\{b_1, \cdots, b_{\xi}\}$.  \qed
\end{proof}

\paragraph{b)}
$\mathsf{Tag}(c_{uv}, k_1, k_2) \rightarrow \{t_1, \cdots, t_m\}$: The client runs this algorithm which inputs $c_{uv}$ and $\{k_1, k_2\}$ where $c_{uv}$ is the $v$-th coded block in the $u$-th server. $c_{uv}$ is computed as a linear combination of $\{w_1, \cdots, w_m\}$ using network coding: $c_{uv} = \sum^m_{i=1} \alpha_i w_i$. The algorithm outputs the tag $t_{uv}$ for $c_{uv}$ such that:
\begin{equation}
t_{uv} = c_{uv} \cdot k_1 + f_{k_2}(u || v) \in \mathbb{F}_q
\end{equation}
for all $i \in \{1, \cdots, m\}$. $u$ denotes the server index and $v$ denotes the coded block index in a server. Suppose the number of servers is $n$ ($u \in \{1, \cdots, n\}$). Suppose the number of coded block in a server is $d$ ($v \in \{1, \cdots, d\}$).

\paragraph{c)}
$\mathsf{Verify} (c_{S_u}, t_{S_u}, k''_1, k_2, S_u) \rightarrow \{0, 1\}$: The verifier runs this algorithm to check a server $S_u$ where $u \in \{1, \cdots, n\}$. The algorithm inputs $\{c_{S_u}, t_{S_u}, k''_1, k_2\}$ where $c_{S_u}$ and $t_{S_u}$ are the linear combinations of $\{c_{u1}, \cdots, c_{ud}\}$ and $\{t_{u1}, \cdots, t_{ud}\}$, respectively. Namely, $c_{S_u} = \sum^d_{v=1} \beta_{uv} c_{uv}$ and $t_{S_u} = \sum^d_{v=1} \beta_{uv} t_{uv}$. The algorithm checks if:
\begin{equation}
\label{equa:verify}
t_{S_u} \stackrel{?}{=} c_{S_u} \cdot k''_1 + \sum^d_{v=1} \beta_{uv} f_{k_2}(u || v)
\end{equation}
This algorithm outputs 1 if the equality holds ($S_u$ is healthy), and outputs 0 otherwise.

The correctness of Eq. \ref{equa:verify} is proved as follows:
\begin{proof}
$ $\linebreak
\begin{tabular}{l l l}
$t_{S_u}$ & = & $ \sum^d_{v=1} \beta_{uv} t_{uv}$\\
& = & $ \sum^d_{v=1} \beta_{uv} (c_{uv} k_1 + f_{k_2}(u || v))$\\
& = & $\sum^d_{v=1} \beta_{uv} c_{uv} k_1 + \sum^d_{v=1} \beta_{uv} f_{k_2}(u || v)$\\
& = & $\sum^d_{v=1} \sum^m_{i=1}  \beta_{uv} \alpha_i w_i k_1 +\sum^d_{v=1} \beta_{uv} f_{k_2}(u||v)$\\
& = & $\sum^d_{v=1} \sum^m_{i=1}  \beta_{uv} \alpha_i w_i (k_1+k'_1) +\sum^d_{v=1} \beta_{uv} f_{k_2}(u||v)$ // because $w_i k'_1 = 0$\\
& = & $\sum^d_{v=1} \beta_{uv} c_{uv}(k_1+k'_1) +\sum^d_{v=1} \beta_{uv} f_{k_2}(u||v)$\\
& = & $c_{S_u}(k_1+k'_1) +\sum^d_{v=1} \beta_{uv} f_{k_2}(u||v)$\\
\end{tabular}
\qed
\end{proof}

\section{Future work}
In the proposed InterCW-MAC, the verifier is given $k'' = k_1 +  k'_1 \in \mathbb{F}^{\xi+m}_q$ and $k_2 \in \mathbb{F}^{\xi+m}_q$. The probability for the verifier to search each secret key $k_1, k_2$ is $1/q^{\xi+m}$. This probability can be reduced if the verifier is given $k''= k_1 +  k'_1 \in \mathbb{F}^{\xi+m}_q$ and $k_2'' = k_2 +  k'_2 \in \mathbb{F}^{\xi+m}_q$. For future work, a potential solution is to design a pseudo-random function $f_2$ such that it is homomorphic (to suit network coding) and such that given $f_2$, $k_2$ and $k'_2$ cannot be obtained.

\end{document}